\documentclass[a4paper]{article}

\usepackage{extarrows} 
\usepackage{amsmath,amsthm,amsfonts,amssymb}
\allowdisplaybreaks[4] 
\newtheorem{theorem}{Theorem}

\newtheorem{lemma}{Lemma}
\newtheorem{proposition}{Proposition}

\newcommand{\op}[1]{\operatorname{#1}}
\newcommand{\bb}[1]{\mathbb{#1}}

\newcommand{\cl}[1]{\mathcal{#1}}
\newcommand{\mbf}[1]{\mathbf{#1}}

\newcommand{\GF}{\mathbb{F}}
\newcommand{\Lra}{\Leftrightarrow}
\newcommand{\ra}{\rightarrow}
\newcommand{\Ra}{\Rightarrow}
\newcommand{\varep}{\varepsilon}
\newcommand{\ord}{\operatorname{ord}}
\newcommand{\wh}[1]{\widehat{#1}}
\newcommand{\ol}[1]{\overline{#1}}

\usepackage{setspace}
\usepackage{indentfirst}
\usepackage{geometry}
\usepackage{paralist}

\usepackage{hyperref}
\hypersetup{colorlinks, linkcolor=blue, bookmarksnumbered=true, bookmarksopen=true}

\title{Multiplicities of Character Values of Binary Sidel'nikov-Lempel-Cohn-Eastman Sequences}
\author{Qi Zhang, Jing Yang}

\begin{document}
\maketitle

\begin{abstract}
Binary Sidel'nikov-Lempel-Cohn-Eastman sequences (or SLCE sequences) over $\GF_2$ have even period and almost
perfect autocorrelation. However, the evaluation of the linear complexity of these sequences is really difficult.
In this paper, we continue the study of \cite{Alaca16}. We first express the multiple roots of character
polynomials of SLCE sequences into certain kinds of Jacobi sums. Then by making use of Gauss sums and Jacobi sums
in the ``semiprimitive'' case, we derive new divisibility results for SLCE sequences.
\end{abstract}


\begin{section}{Introduction}
Let $\GF$ be a finite field and $\mbf{s}=(s_0,s_1,s_2,\cdots)$ be a sequence over $\GF$. $\mbf{s}$ is
called periodic if there exists a positive integer $T$ such that $s_i=s_{i+T}$ for any $i\geq0$. The smallest
such $T$ is called the period of $\mbf{s}$.

Periodic sequences applied in stream ciphers should has some good properties, such as good balance property,
low autocorrelation, large period and large linear complexity \cite{Golomb04}. In this paper, we focus on
linear complexity.

For a sequence $\mbf{s}$, if there exists a positive integer $L$ and $c_1,c_2,\cdots,c_L\in\GF$ such that
\[s_i=-(c_1s_{i-1}+c_2s_{i-2}+\cdots+c_Ls_{i-L})\ \ \ \mbox{for each }i\geq L\]
then the smallest such $L$ is called the linear complexity of $\mbf{s}$, and corresponding polynomial
$c(X):=1+c_1X+c_2X^2+\cdots+c_LX^L$ is called the minimal polynomial of $\mbf{s}$.

Throughout the paper, let $p$ denote a prime number, $q=p^m$, $\GF_q$ the finite field of order $q$, $\alpha$
a primitive element of $\GF_q$, $d\mid(q-1)$ a prime number. Since $\GF_q^*$ is a cyclic group of order $q-1$,
$\langle\alpha^d\rangle$ is the unique subgroup of index $d$. We define $C_i:=\alpha^i\langle\alpha^d\rangle$
as the $i$th cyclotomic coset of index $d$. Then the $d$-ary SLCE sequence $\mbf{s}=(s_0,s_1,s_2,\cdots)$ over
$\GF_d$ can be defined by
\[s_n=\begin{cases}
i & \mbox{if }\alpha^n+1\in C_i\mbox{ for some }i \\
0 & \mbox{otherwise}
\end{cases}\]

SLCE sequences were introduced by Sidel'nikov \cite{Sidel69} and Lempel, Cohn and Eastman \cite{LCE77}
(the binary case, i.e., $d=2$)\footnote{$d=2$ implies $p$ is a odd prime number.} independently.
The sequence $\mbf{s}$ has even period $T=q-1$, perfect balance property and almost optimal autocorrelation
\cite{Sidel69,LCE77}. A natural problem is to study the linear complexity of SLCE sequences.

Let $S(X)=s_0+s_1X+s_2X+\cdots+s_{T-1}X^{T-1}$. It is a basic result \cite{Golomb04} that the minimal polynomial
of $\mbf{s}$ is given by
\[c(X)=\frac{X^T-1}{\gcd(X^T-1,S(X))}\]
Therefore, the linear complexity of $\mbf{s}$ is given by
\[L=\deg c(X)=T-\deg(\gcd(X^T-1,S(X)))\]
So it is enough to decide $\gcd(X^T-1,S(X))$. We emphasize that these are polynomials over $\GF_d$ whose
characteristic is a divisor of $T$, which means we should not only determine the common roots but also compute
the multiplicities of these roots.

Let $\beta$ be a root of $X^T-1$. Then $\beta$ is a $k$th primitive root of unity over $\GF_d$, where $k\mid T$
and $d\nmid k$. Let $I_\beta(X)$ denote the minimal polynomial of $\beta$ over $\GF_d$. We need to determine
whether $I_\beta(X)\mid S(X)$, i.e., $S(\beta)=0$. Furthermore, we have to compute the multiplicities. This
results are called ``divisibility results'' \cite{Alaca16}.

\cite{Helle02}, \cite{Pott03}, \cite{MeidlF206}, and \cite{Alaca16} studied the linear complexity of binary
SLCE sequences. \cite{MeidlFd06} studied the linear complexity of $d$-ary SLCE sequences.
\cite{Meidl08} generalized the definition of SLCE sequences to the case $d$ is a prime power and studied
the linear complexity of generalized SLCE sequences over nonprime finite fields. Since it seems difficult to
compute the linear complexity of $d$-ary SLCE sequences over $\GF_d$, some researchers studied the linear
complexity and $k$-error linear complexity of SLCE sequences over $\GF_p$. We refer the reader to
\cite{HelleFp03}, \cite{Helle04}, \cite{Garaev06}, \cite{Eun05}, \cite{Kim06}, \cite{Kim05}, \cite{Winter06},
\cite{Chung06} and \cite{Meidl07}.

In this paper, we continue the study of the linear complexity of binary SLCE sequences in \cite{Alaca16}.
\cite{Alaca16} gave a necessary and sufficient condition of $S(\beta)=0$, which is a congruence of
Jacobi sums. Since Jacobi sums can be evaluated in some specific cases, authors of \cite{Alaca16} were able to
give new divisibility results. In this paper, we focus on the multiplicities of $\beta$. By making use of
Hasse derivative and Lucas's congruence, we can generalize the condition in \cite{Alaca16}.
Finally, we prove that a root must be multiple root in semiprimitive case, which is a new divisibility result.

This paper is organized as follows. Section \ref{sec:pre} briefly reviews the mathematical tools which are useful
in the sequel. Section \ref{sec:main} gives the main results. Section \ref{sec:semi} applies the main results to
semiprimitive case. Section \ref{sec:con} concludes this paper.
\end{section}

\begin{section}{Preliminaries}\label{sec:pre}
\begin{subsection}{Multiplicative Characters}
A multiplicative character $\chi$ of a finite field $\GF_q$ is defined to be a group homomorphism from
$\GF_q^*$ to $\bb{C}^*$. Since $\GF_q^*=\langle\alpha\rangle$, $\chi$ is completely determined by $\chi(\alpha)$.
It is easy to see $\chi(\alpha)^{q-1}=\chi(\alpha^{q-1})=\chi(1)=1$, so $\chi(\alpha)$ must be a $(q-1)$th
root of unity, i.e., the image of $\chi$ is $\langle\zeta_{q-1}\rangle$, where $\zeta_n$ denotes
$\exp(2\pi\sqrt{-1}/n)$.

The set of all multiplicative characters of $\GF_q$ is denoted by $\wh{\GF_q^*}$. It is well-known that
\[\wh{\GF_q^*}=\{\eta_{\frac{i}{q-1}}:0\leq i\leq q-1\}\]
where $\eta_j$ is determined by $\eta_j(\alpha)=\exp(2\pi\sqrt{-1}j)$ for rational number $j$.
Moreover, $\wh{\GF_q^*}$ forms a group isomorphic to $\GF_q^*$ under the point-wise multiplication.
The identity element $\varep$ is defined by $\varep(x)=1$ for each $x\in\GF_q^*$.
The inverse element of $\chi$ is $\ol\chi$, where $\ol\chi(x)=\ol{\chi(x)}$ for each $x\in\GF_q^*$.

For convenience, we extend the domain of a multiplicative character $\chi$ to $\GF_q$ by setting $\chi(0):=0$.

We need following results of multiplicative characters.
\begin{lemma}\label{lem:character1}
For $d\mid(q-1)$ and $x\in\GF_q^*$, we have
\[\frac{1}{d}\sum_{i=0}^{d-1}\eta_{\frac{i}{d}}(x)=\begin{cases}
1 & \mbox{if }x\in\langle\alpha^d\rangle \\
0 & \mbox{if }x\not\in\langle\alpha^d\rangle
\end{cases}\]
\end{lemma}
\begin{lemma}\label{lem:character2}
Let $d_1,d_2\mid(q-1)$, $\chi$ a multiplicative character of order $d_1$. If $\gcd(d_1,d_2)=1$ and $d_1>1$, then
\[\sum_{x\in\langle\alpha^{d_2}\rangle}\chi(x)=0\]
Furthermore, for each $i=0,1,\cdots,d_2-1$
\[\sum_{x\in\alpha^i\langle\alpha^{d_2}\rangle}\chi(x)=0\]
\end{lemma}
\end{subsection}
\begin{subsection}{Gauss sums}
Let $\chi$ be a multiplication character of $\GF_q$. The Gauss sum is defined as follows
\[G(\chi):=\sum_{x\in\GF_q}\chi(x)\zeta_p^{\op{Tr}(x)}\]
where $\op{Tr}$ is the trace mapping from $\GF_q$ to $\GF_p$.

Gauss sums are extremely difficult to evaluate in general. But there are several special cases in which we
can get the explicit formula for Gauss sums. We need Gauss sums in quadratic case and semiprimitive case
in the sequel.

\begin{lemma}[\cite{GJSum}, Theorem 11.5.4]\label{lem:quadr}
Let $\rho=\eta_{\frac{1}{2}}$ be the quadratic multiplicative character of $\GF_q$. Then
\[G(\rho)=\begin{cases}
(-1)^{m-1}\sqrt{q} & \mbox{if }p\equiv1\pmod4 \\
(-1)^{m-1}i^m\sqrt{q} & \mbox{if }p\equiv3\pmod4
\end{cases}\]
\end{lemma}
\begin{lemma}[\cite{GJSum}, Theorem 11.6.3]\label{lem:semi}
Let $\chi$ be a multiplicative character of $\GF_q$, $N=\op{ord}(\chi)>2$. Suppose there exists a positive integer
$v$ such that $p^v\equiv-1\pmod{N}$, with $v$ chosen minimal. Then $m=2vw$ for some positivi integer $w$ and
\[G(\chi)=(-1)^{w-1+pw\frac{p^v+1}{N}}\]
\end{lemma}
\end{subsection}
\begin{subsection}{Jacobi sums}
Let $\chi_1,\chi_2$ be multiplicative characters of $\GF_q$. The Jacobi sum is defined by
\[J(\chi_1,\chi_2):=\sum_{x\in\GF_q}\chi_1(x)\chi_2(1-x)\]
We also denote $J(\rho,\chi)$ by $K(\chi)$ where $\rho$ is a quadratic character.

Jacobi sums are closely related to Gauss sums according to following lemma.
\begin{lemma}[\cite{GJSum}, Theorem 2.1.3]\label{lem:Jacobi}
If $\chi_1\chi_2\neq\varep$, then
\[J(\chi_1,\chi_2)=\frac{G(\chi_1)G(\chi_2)}{G(\chi_1\chi_2)}\]
\end{lemma}
\end{subsection}
\begin{subsection}{Cyclotomic Fields}
Let $k$ be an odd positive integer, $K=\bb{Q}(\zeta_k)$ the $k$th cyclotomic field, $\cl{O}_K=\bb{Z}[\zeta_k]$
the ring of algebraic integers in $K$. We need some results about the prime ideal factorization of $(2)$ in
$\cl{O}_K$.
\begin{lemma}[\cite{NumTheory}]\label{lem:ideal}
Let $f=\ord_k(2)$, i.e., $f$ is the smallest positive integer such that $k\mid 2^f-1$. Let $\cl{P}$ be any prime
ideal of $\cl{O}_K$ over $2$. We have
\begin{compactitem}
  \item $\cl{O}_K/\cl{P}\cong\GF_{2^f}$.
  \item $1,\zeta_k,\cdots,\zeta_k^{k-1}$ are mutually distinct modulo $\cl{P}$.
  \item $\forall\,\gamma\in\cl{O}_K-\cl{P}$, there exists a unique $k$th root of unity $\zeta$ such that
    $\gamma^{\frac{2^f-1}{k}}\equiv\zeta\pmod{\cl{P}}$
\end{compactitem}
\end{lemma}
\end{subsection}
\begin{subsection}{Hasse Derivative}
Hasse derivative is a powerful tool to study the multiplicities of a root in a polynomial over a field of
finite characteristic. Let $f(X)=a_0+a_1X+\cdots+a_nX^n\in\GF[X]$, where $\GF$ is a finite field.
The $t$th Hasse derivative of $f(X)$ is defined by
\[f(X)^{(t)}:=\binom{t}{t}a_t+\binom{t+1}{t}a_{t+1}X+\cdots+\binom{n}{t}a_nX^{n-t}\]
Hasse derivative has following similar property to normal derivative.
\begin{lemma}[\cite{FiniteF}, Lemma 6.51]\label{lem:Hasse}
Let $f(X)\in\GF[X]$ be a nonzero polynomial, $\xi\in\GF$. Then the multiplicities of $\xi$ in $f(X)$
is $t$ if and only if $f(\xi)^{(0)}=f(\xi)^{(1)}=\cdots=f(\xi)^{(t-1)}=0$, $f(\xi)^{(t)}\neq0$.
\end{lemma}
\end{subsection}
\begin{subsection}{Lucas's Congruence}
Lucas's congruence deals with binomial coefficients modulo a prime number.
\begin{lemma}[\cite{Lucas}]\label{lem:Lucas}
Let $n,k$ be positive integers. Then
\[\binom{n}{k}\equiv\binom{n_0}{k_0}\binom{n_1}{k_1}\cdots\binom{n_l}{k_l}\pmod{2}\]
where $n=n_0+n_12+\cdots+n_l2^l$ and $k=k_0+k_12+\cdots+k_l2^l$ are $2$-adic representation of $n$ and $k$,
respectively.
\end{lemma}
\end{subsection}
\end{section}

\begin{section}{Main Results}\label{sec:main}
From now on, we always assume $d=2$, which implies $p$ is odd. Let $\beta$ be a root of $X^T-1$ over $\GF_2$,
$k=\ord(\beta)$, $f$ the smallest positive integer such that $k\mid(2^f-1)$. Then $k$ is odd divisor of $T=q-1$
and $\GF_2(\beta)=\GF_{2^f}$. We further assume $k>1$ since $k=1$ implies $\beta=1$ which was already studied
in details \cite{MeidlF206}. Let $T=q-1=2^uT'$ where $u\geq1$ and $T'$ is odd. Then $\beta$ is $2^u$th multiple
root of $X^T-1$. Our aim is to determine the multiplicities of $\beta$ in $S(X)$. According to lemma
\ref{lem:Hasse}, it is enough to decide whether $S(\beta)^{(t)}=0$ for $0\leq t\leq 2^u-1$.

\begin{theorem}\label{thm:1}
Let notations be as above. For any $t\in\{0,1,\cdots,2^u-1\}$,
\[S(\beta)^{(t)}=0\ \ \Lra\ \ \binom{T/2}{t}+\sum_{n=0}^{T-1}\binom{n}{t}\rho(\alpha^n+1)\chi(\alpha^n)\equiv0\pmod{2P}\]
\end{theorem}
\begin{proof}
Let $K=\bb{Q}(\zeta_k)$, $\cl{O}_K=\bb{Z}[\zeta_k]$ and $\cl{P}$ be a stipulated prime ideal of $K$ over $2$.
By lemma \ref{lem:ideal}, $\GF_2(\beta)=\GF_{2^f}\cong\cl{O}_K/\cl{P}$. Thus we can introduce a field isomorphism
$\phi:\GF(\beta)\ra\cl{O}_K/\cl{P}$ and
\begin{align*}
S(\beta)^{(t)}=0\ \Lra\ & \beta^tS(\beta)^{(t)}=0\ \Lra\ \sum_{n=0}^{T-1}\binom{n}{t}s_n\beta^n=0 \\
\Lra\ & \phi(\sum_{n=0}^{T-1}\binom{n}{t}s_n\beta^n)=0\ \Lra\ \sum_{n=0}^{T-1}\binom{n}{t}s_n\phi(\beta)^n=0
\end{align*}
Notice that $\beta$ has order $k$, i.e., $\beta=\theta^{\frac{2^f-1}{k}}$ for some $\theta\in\GF_{2^f}^*$.
It follows from lemma \ref{lem:ideal} that there exists a unique $k$th primitive root of unity $\zeta$ such that
$\phi(\beta)\equiv\zeta\pmod{\cl{P}}$. Define the function $\chi:\GF_q^*\ra K^*$ by $\chi(\alpha^i):=\zeta^i$
for arbitrary $i$. It is simple to verify that $\chi$ is a multiplicative character of $\GF_q$ and $\ord(\chi)=k$.
Then we have
\begin{align*}
S(\beta)^{(t)}=0\ \Lra\ & \sum_{n=0}^{T-1}\binom{n}{t}s_n\zeta^n\equiv0\pmod{\cl{P}} \\
\Lra\ & \sum_{n=0}^{T-1}\binom{n}{t}s_n\chi(\alpha^n)\equiv0\pmod{\cl{P}}
\end{align*}
Recall that $s_n=\begin{cases}
1 & \mbox{if }\alpha^n+1\in\alpha\langle\alpha^2\rangle \\
0 & \mbox{if }\alpha^n+1\not\in\alpha\langle\alpha^2\rangle
\end{cases}$. We introduce the quadratic character $\rho=\eta_{\frac{1}{2}}$ and $s_n$ can be expressed as
$\frac{1-\rho(\alpha^n+1)-\delta_{n,T/2}}{2}$. Thus,
\begin{align*}
S(\beta)^{(t)}=0\ \Lra\ & \sum_{n=0}^{T-1}\binom{n}{t}2s_n\chi(\alpha^n)\equiv0\pmod{2\cl{P}} \\
\Lra\ & \sum_{n=0}^{T-1}\binom{n}{t}(1-\rho(\alpha^n+1)-\delta_{n,T/2})\chi(\alpha^n)\equiv0\pmod{2\cl{P}} \\
\Lra\ & \sum_{n=0}^{T-1}\binom{n}{t}\chi(\alpha^n)-\sum_{n=0}^{T-1}\binom{n}{t}\rho(\alpha^n+1)\chi(\alpha^n)-\binom{T/2}{t}\chi(\alpha^{T/2})\equiv0\pmod{2\cl{P}}
\end{align*}
Observe that $\binom{n}{t}\equiv1\pmod2$ if and only if $n$ is congruent to some specific numbers modulo $2^h$
where $h$ is the length of binary representation of $t$ by lemma \ref{lem:Lucas}. The first summation can be
divided into several summations. Each summation has the form $\sum_{n\equiv i\pmod{2^h}}\chi(\alpha^n)$.
Since $t\leq 2^u-1$, we have $2^h\mid 2^u$ and $\sum_{n\equiv i\pmod{2^h}}\chi(\alpha^n)=
\sum_{x\in\alpha^i\langle\alpha^{2^h}\rangle}\chi(x)$, which is equal to $0$ due to lemma \ref{lem:character2}.
On the other hand, $\chi(\alpha^{T/2})=\chi(-1)=1$ follows from the fact $\chi$ has odd order.
Therefore,
\[S(\beta)^{(t)}=0\ \Lra\ -\sum_{n=0}^{T-1}\binom{n}{t}\rho(\alpha^n+1)\chi(\alpha^n)-\binom{T/2}{t}\equiv0\pmod{2\cl{P}}\]
which completes the proof.
\end{proof}

As we already mentioned in the proof, the summation $\sum_{n=0}^{T-1}\binom{n}{t}\rho(\alpha^n+1)\chi(\alpha^n)$
can be divided into several summations, where each is summed over a cyclotomic coset. Let
$h=\begin{cases}
0 & \mbox{if }t=0 \\
\min\{i\geq1:t\leq 2^i\} & \mbox{if }t>0
\end{cases}$ be the length of binary representation of $t$ and
\[I_t=\{0\leq i\leq 2^h-1:i\succeq t\}\]
where $a\succeq b$ means each binary bit of $a$ is great than or equal to that of $b$, i.e.,
$\binom{a}{b}\equiv1\pmod{2}$ due to Lucas's congruence. Note that, $\binom{n}{t}\equiv1\pmod2$ if and only if
$n$ is congruent to some number in $I_t$ modulo $2^h$. Now we can use $I_t$ and lemma \ref{lem:character1}
to write previous theorem into an expression involving Jacobi sums.

\begin{theorem}\label{thm:2}
Under the notations above, $S(\beta)^{(t)}=0$ is equivalent to
\[2^h\binom{T/2}{t}+\sum_{i\in I_t}\sum_{j=0}^{2^h-1}\eta_{\frac{j}{2^h}}(-1)\zeta_{2^h}^{-ij}K(\eta_{\frac{j}{2^h}}\chi)\equiv0\pmod{2^{h+1}\cl{P}\bb{Z}[\zeta_{2^hk}]}\]
\end{theorem}
\begin{proof}
For a proposition $Q$, put
\[\delta_Q=\begin{cases}
1 & \mbox{if }Q\mbox{ is true} \\
0 & \mbox{if }Q\mbox{ is false}
\end{cases}\]
It follows that
\begin{align*}
S(\beta)^{(t)}=0\ \Lra\ & \binom{T/2}{t}+\sum_{n=0}^{T-1}\binom{n}{t}\rho(\alpha^n+1)\chi(\alpha^n)\equiv0\pmod{2\cl{P}} \\
\Lra\ & \binom{T/2}{t}+\sum_{n=0}^{T-1}\sum_{i\in I_t}\delta_{n\equiv i(\op{mod}2^h)}\rho(\alpha^n+1)\chi(\alpha^n)\equiv0\pmod{2\cl{P}} \\
\Lra\ & \binom{T/2}{t}+\sum_{i\in I_t}\sum_{n=0}^{T-1}\delta_{\alpha^{n-i}\in\langle\alpha^{2^h}\rangle}\rho(\alpha^n+1)\chi(\alpha^n)\equiv0\pmod{2\cl{P}} \\
\Lra\ & 2^h\binom{T/2}{t}+\sum_{i\in I_t}\sum_{n=0}^{T-1}2^h\delta_{\alpha^{n-i}\in\langle\alpha^{2^h}\rangle}\rho(\alpha^n+1)\chi(\alpha^n)\equiv0\pmod{2^{h+1}\cl{P}\bb{Z}[\zeta_{2^h}k]} \\
\Lra\ & 2^h\binom{T/2}{t}+\sum_{i\in I_t}\sum_{n=0}^{T-1}\sum_{j=0}^{2^h-1}\eta_{\frac{j}{2^h}}(\alpha^{n-i})\rho(\alpha^n+1)\chi(\alpha^n)\equiv0\pmod{2^{h+1}\cl{P}\bb{Z}[\zeta_{2^h}k]} \\
\Lra\ & 2^h\binom{T/2}{t}+\sum_{i\in I_t}\sum_{j=0}^{2^h-1}\zeta_{2^h}^{-ij}\sum_{n=0}^{T-1}\rho(\alpha^n+1)\eta_{\frac{j}{2^h}}\chi(\alpha^n)\equiv0\pmod{2^{h+1}\cl{P}\bb{Z}[\zeta_{2^h}k]} \\
\Lra\ & 2^h\binom{T/2}{t}+\sum_{i\in I_t}\sum_{j=0}^{2^h-1}\zeta_{2^h}^{-ij}\sum_{x\in\GF_q}\rho(x+1)\eta_{\frac{j}{2^h}}\chi(x)\equiv0\pmod{2^{h+1}\cl{P}\bb{Z}[\zeta_{2^h}k]} \\
\Lra\ & 2^h\binom{T/2}{t}+\sum_{i\in I_t}\sum_{j=0}^{2^h-1}\eta_{\frac{j}{2^h}}\chi(-1)\zeta_{2^h}^{-ij}\sum_{x\in\GF_q}\rho(-x+1)\eta_{\frac{j}{2^h}}\chi(x)\equiv0\pmod{2^{h+1}\cl{P}\bb{Z}[\zeta_{2^h}k]} \\
\Lra\ & 2^h\binom{T/2}{t}+\sum_{i\in I_t}\sum_{j=0}^{2^h-1}\eta_{\frac{j}{2^h}}(-1)\zeta_{2^h}^{-ij}J(\rho,\eta_{\frac{j}{2^h}}\chi)\equiv0\pmod{2^{h+1}\cl{P}\bb{Z}[\zeta_{2^h}k]}
\end{align*}
It completes the proof.
\end{proof}

For specific $t$, one can obtain explicit condition by theorem \ref{thm:2}.
We give several examples to show this and omit the details of the proof.

\begin{proposition}[\cite{Alaca16}, Theorem 3.1]\label{prop:1}
Keep the notations as above. $S(\beta)=0$ is equivalent to $1+K(\chi)\equiv0\pmod{2\cl{P}}$.
\end{proposition}
\begin{proposition}
Keep the notations as above. $S(\beta)^{(1)}=0$ is equivalent to
\[\begin{cases}
K(\chi)-K(\eta_{\frac{1}{2}}\chi)\equiv0\pmod{4\cl{P}} & \mbox{if }q\equiv1\pmod4 \\
2+K(\chi)+K(\eta_{\frac{1}{2}}\chi)\equiv0\pmod{4\cl{P}} & \mbox{if }q\equiv3\pmod4
\end{cases}\]
\end{proposition}
\begin{proposition}
Keep the notations as above. If $q\equiv1\pmod4$, then $S(\beta)^{(2)}=0$ is equivalent to
\[\begin{cases}
2K(\chi)-(1-\zeta_4)K(\eta_{\frac{1}{4}}\chi)-(1+\zeta_4)K(\eta_{\frac{3}{4}}\chi)\equiv0\pmod{8\cl{P}\bb{Z}[\zeta_{4k}]} & \mbox{if }q\equiv1\pmod8 \\
4+2K(\chi)+(1-\zeta_4)K(\eta_{\frac{1}{4}}\chi)+(1+\zeta_4)K(\eta_{\frac{3}{4}}\chi)\equiv0\pmod{8\cl{P}\bb{Z}[\zeta_{4k}]} & \mbox{if }q\equiv5\pmod8
\end{cases}\]
\end{proposition}
\begin{proposition}
Keep the notations as above. If $q\equiv1\pmod4$, then $S(\beta)^{(3)}=0$ is equivalent to
\[\begin{cases}
K(\chi)+\zeta_4K(\eta_{\frac{1}{4}}\chi)-K(\eta_{\frac{1}{2}}\chi)-\zeta_4K(\eta_{\frac{3}{4}}\chi)\equiv0\pmod{8\cl{P}\bb{Z}[\zeta_{4k}]} & \mbox{if }q\equiv1\pmod8 \\
K(\chi)-\zeta_4K(\eta_{\frac{1}{4}}\chi)-K(\eta_{\frac{1}{2}}\chi)+\zeta_4K(\eta_{\frac{3}{4}}\chi)\equiv0\pmod{8\cl{P}\bb{Z}[\zeta_{4k}]} & \mbox{if }q\equiv5\pmod8
\end{cases}\]
\end{proposition}

It seems difficult to give a simple criteria to determine whether $\beta$ is $t$th multiple root of $S(X)$
for general $t$. But we can obtain a relatively simple condition when $t=2^h$.

\begin{theorem}{\label{thm:3}}
With the notations above, $\beta$ has multiplicities at least $2^h$ as root of $S(X)$ if and only if
\[\begin{bmatrix}
1 & 1 & \cdots & 1 \\
1 & \zeta_{2^h}^{-1} & \cdots & \zeta_{2^h}^{-(2^h-1)} \\
\vdots & \vdots & \ddots & \vdots \\
1 & \zeta_{2^h}^{-(2^h-1)} & \cdots & \zeta_{2^h}^{-(2^h-1)(2^h-1)}
\end{bmatrix}
\begin{bmatrix}
\eta_{\frac{0}{2^h}}(-1)K(\eta_{\frac{0}{2^h}}\chi) \\
\eta_{\frac{1}{2^h}}(-1)K(\eta_{\frac{1}{2^h}}\chi) \\
\vdots \\
\eta_{\frac{2^h-1}{2^h}}(-1)K(\eta_{\frac{2^h-1}{2^h}}\chi)
\end{bmatrix}\equiv-2^h
\begin{bmatrix}\delta_0 \\ \delta_1 \\ \vdots \\ \delta_{2^h-1}\end{bmatrix}
\pmod{2^{h+1}\cl{P}\bb{Z}[\zeta_{2^hk}]}\]
where $\delta_i=\begin{cases}
1 & \mbox{if }T/2\equiv i\pmod{2^h} \\
0 & \mbox{if }T/2\not\equiv i\pmod{2^h}
\end{cases}$.
\end{theorem}
\begin{proof}
Consider the summation
\[E_i=\sum_{\substack{n=0 \\ n\equiv i(\op{mod}2^h)}}^{T-1}s_n\beta^n\]
Recall that $S(\beta)^{(t)}=0$ if and only if $\sum_{n=0}^{T-1}\binom{n}{t}s_n\beta^n=0$.
The summation of the left hand side can break up into several $E_i$, and the equation becomes
\begin{equation}
\sum_{\substack{i=0 \\ i\succeq t}}^{2^h-1}E_i=0
\tag{$D_t$}
\end{equation}
From lemma \ref{lem:Hasse}, $\beta$ has multiplicities at least $2^h$ if and only if $D_t$ holds for
every $t\in\{0,1,\cdots,2^h-1\}$.

Next, we prove that equations $D_t$ where $t=0,1,\cdots,2^h-1$ are equivalent to equations $E_i=0$
where $i=0,1,\cdots,2^h-1$. It is trivial that equations $E_i$s implies $D_t$s.
We only need to prove that $E_i$s can be derived from $D_t$s. Notice that $D_{2^h-1}$ is actually the same as
$E_{2^h-1}=0$. Since these equations are all over a field of characteristic $2$, we can add $D_{2^h-1}$ to
other $D_t$s to cancel the summation $E_{2^h}-1$ occurring in other $D_t$s.
After the process, the equation $D_{2^h-2}$ becomes $E_{2^h-2}=0$ and we can use the same method to cancel
$E_{2^h-2}$ occurring in other $D_t$s where $0\leq t\leq 2^h-3$. Then $D_{2^h-3}$ becomes $E_{2^h-3}=0$.
Repeat these elementary operations until $D_1$ becomes $E_1=0$ and it completes the proof.

Finally, we use the same techniques in proof of theorem \ref{thm:1} and \ref{thm:2}
\begin{align*}
E_i=0\ \Lra\ & \sum_{n=0}^{T-1}\delta_{n\equiv i(\op{mod}2^h)}s_n\beta^n=0 \\
\Lra\ & \sum_{n\equiv i(\op{mod}2^h)}s_n\chi(\alpha^n)\equiv0\pmod{\cl{P}} \\
\Lra\ & \sum_{n\equiv i(\op{mod}2^h)}2s_n\chi(\alpha^n)\equiv0\pmod{2\cl{P}} \\
\Lra\ & \sum_{n\equiv i(\op{mod}2^h)}(1-\rho(\alpha^n+1)-\delta_{n,T/2})\chi(\alpha^n)\equiv0\pmod{2\cl{P}} \\
\Lra\ & \sum_{n\equiv i(\op{mod}2^h)}(-\rho(\alpha^n+1)-\delta_{n,T/2})\chi(\alpha^n)\equiv0\pmod{2\cl{P}} \\
\Lra\ & \delta_i+\sum_{n\equiv i(\op{mod}2^h)}\rho(\alpha^n+1)\chi(\alpha^n)\equiv0\pmod{2\cl{P}} \\
\Lra\ & 2^h\delta_i+2^h\sum_{n\equiv i(\op{mod}2^h)}\rho(\alpha^n+1)\chi(\alpha^n)\equiv0\pmod{2^{h+1}\cl{P}\bb{Z}[\zeta_{2^hk}]} \\
\Lra\ & 2^h\delta_i+\sum_{n=0}^{T-1}\sum_{j=0}^{2^h-1}\eta_{\frac{j}{2^h}}(\alpha^{n-i})\rho(\alpha^n+1)\chi(\alpha^n)\equiv0\pmod{2^{h+1}\cl{P}\bb{Z}[\zeta_{2^hk}]} \\
\Lra\ & \sum_{j=0}^{2^h-1}\zeta_{2^h}^{-ij}\sum_{x\in\GF_q}\rho(x+1)\eta_{\frac{j}{2^h}}\chi(x)\equiv-2^h\delta_i\pmod{2^{h+1}\cl{P}\bb{Z}[\zeta_{2^hk}]} \\
\Lra\ & \sum_{j=0}^{2^h-1}\zeta_{2^h}^{-ij}\eta_{\frac{j}{2^h}}(-1)K(\eta_{\frac{j}{2^h}}\chi)\equiv-2^h\delta_i\pmod{2^{h+1}\cl{P}\bb{Z}[\zeta_{2^hk}]}
\end{align*}
It completes the proof.
\end{proof}

\begin{theorem}
With the notations above, a necessary condition for $\beta$ having multiplicities at least $2^h$ is
$1+K(\eta_{\frac{j}{2^h}}\chi)\equiv0\pmod{2\cl{P}\bb{Z}[\zeta_{2^hk}]}$ for any $j\in\{0,1,\cdots,2^h-1\}$.
\end{theorem}
\begin{proof}
Let
\[C=\begin{bmatrix}
1 & 1 & \cdots & 1 \\
1 & \zeta_{2^h}^{-1} & \cdots & \zeta_{2^h}^{-(2^h-1)} \\
\vdots & \vdots & \ddots & \vdots \\
1 & \zeta_{2^h}^{-(2^h-1)} & \cdots & \zeta_{2^h}^{-(2^h-1)(2^h-1)}
\end{bmatrix}\]
Observing $C^*C=2^hI$, left multiply $C^*$ to the congruence in theorem \ref{thm:3} and we obtain
\begin{align*}
& 2^h\begin{bmatrix}\eta_{\frac{0}{2^h}}(-1)K(\eta_{\frac{0}{2^h}}\chi) \\ \eta_{\frac{1}{2^h}}(-1)K(\eta_{\frac{1}{2^h}}\chi) \\ \vdots \\ \eta_{\frac{2^h-1}{2^h}}(-1)K(\eta_{\frac{2^h-1}{2^h}}\chi)\end{bmatrix}
\equiv-2^hC^*\begin{bmatrix}\delta_0 \\ \delta_1 \\ \vdots \\ \delta_{2^h-1}\end{bmatrix}
\pmod{2^{h+1}\cl{P}\bb{Z}[\zeta_{2^hk}]} \\
\Ra\ & \begin{bmatrix}\eta_{\frac{0}{2^h}}(-1)K(\eta_{\frac{0}{2^h}}\chi) \\ \eta_{\frac{1}{2^h}}(-1)K(\eta_{\frac{1}{2^h}}\chi) \\ \vdots \\ \eta_{\frac{2^h-1}{2^h}}(-1)K(\eta_{\frac{2^h-1}{2^h}}\chi)\end{bmatrix}
\equiv-C^*\begin{bmatrix}\delta_0 \\ \delta_1 \\ \vdots \\ \delta_{2^h-1}\end{bmatrix}
\pmod{2^\cl{P}\bb{Z}[\zeta_{2^hk}]}
\end{align*}
If $h=u$, then $\eta_{\frac{1}{2^h}}(-1)=-1$ and $\delta_i=\delta_{i,2^{h-1}}$.
If $h<u$, then $\eta_{\frac{1}{2^h}}(-1)=1$ and $\delta_i=\delta_{i,0}$.
In any of these two cases, the congruence can be simplified to
\[\begin{bmatrix}K(\eta_{\frac{0}{2^h}}\chi) \\ K(\eta_{\frac{1}{2^h}}\chi) \\ \vdots \\ K(\eta_{\frac{2^h-1}{2^h}}\chi)\end{bmatrix}
\equiv\begin{bmatrix}-1 \\ -1 \\ \vdots \\ -1\end{bmatrix}
\pmod{2^\cl{P}\bb{Z}[\zeta_{2^hk}]}\]
It completes the proof.
\end{proof}
\end{section}

\begin{section}{Divisibility Results in Semiprimitive Case}\label{sec:semi}
In this section, we apply the theorem we get to semiprimitive case. Throughout this section, we assume that
there exists a positive integer $v$ such that $2^hk\mid(p^v+1)$. It follows that $m=2vw$ for some $w$.
In addition, it must be the case $h<u$ because $2^h\mid p^v+1$ and $\frac{p^m-1}{p^v+1}$ is even.

\begin{lemma}\label{lem:1}
With the notations and assumptions above, $K(\eta_{\frac{i}{2^h}}\chi)=G(\rho)$ for any $i\in\{0,1,\cdots,2^h-1\}$.
\end{lemma}
\begin{proof}
Using lemma \ref{lem:Jacobi}, we have
\[K(\eta_{\frac{i}{2^h}}\chi)=J(\rho,\eta_{\frac{i}{2^h}}\chi)=\frac{G(\rho)G(\eta_{\frac{i}{2^h}}\chi)}{G(\rho\eta_{\frac{i}{2^h}}\chi)}\]
We only need to prove $G(\eta_{\frac{i}{2^h}}\chi)=G(\rho\eta_{\frac{i}{2^h}}\chi)$.

If $i=0$ or $i=2^{h-1}$, then we need to prove $G(\chi)=G(\rho\chi)$. Note that both two are semiprimitive Gauss
sum. The smallest positive integers $v_1,v_2$ such that $k\mid(p^{v_1}+1),2k\mid(p^{v_2}+1)$ are equal.
Hence, $G(\chi)=G(\rho\chi)$ follows from lemma \ref{lem:semi}.

Otherwise, the order of $\eta_{\frac{i}{2^h}}\chi$ is the equal to that of $\rho\eta_{\frac{i}{2^h}}\chi$.
Thus it trivially follows from lemma \ref{lem:semi}.
\end{proof}

\begin{lemma}\label{lem:2}
With the notations and assumptions above, if $\beta$ is a root of $S(X)$, then the multiplicities of $\beta$
is $\geq2^h$.
\end{lemma}
\begin{proof}
Notice that $h<u$, which implies $\eta_{\frac{1}{2^h}}(-1)=1$ and $\delta_i=\delta_{i,0}$.
By theorem \ref{thm:3}, $\beta$ has multiplicities $\geq2^h$ if and only if
\[\begin{bmatrix}
1 & 1 & \cdots & 1 \\
1 & \zeta_{2^h}^{-1} & \cdots & \zeta_{2^h}^{-(2^h-1)} \\
\vdots & \vdots & \ddots & \vdots \\
1 & \zeta_{2^h}^{-(2^h-1)} & \cdots & \zeta_{2^h}^{-(2^h-1)(2^h-1)}
\end{bmatrix}
\begin{bmatrix}K(\eta_{\frac{0}{2^h}}\chi) \\ K(\eta_{\frac{1}{2^h}}\chi) \\ \vdots \\ K(\eta_{\frac{2^h-1}{2^h}}\chi)\end{bmatrix}
\equiv-2^h\begin{bmatrix}1 \\ 0 \\ \vdots \\ 0\end{bmatrix}
\pmod{2^{h+1}\cl{P}\bb{Z}[\zeta_{2^hk}]}\]
From lemma \ref{lem:1}, we know $K(\eta_{\frac{i}{2^h}}\chi)=K(\chi)=G(\rho)$ for arbitrary $i$. Thus,
\[\op{l.h.s}=
\begin{bmatrix}
1 & 1 & \cdots & 1 \\
1 & \zeta_{2^h}^{-1} & \cdots & \zeta_{2^h}^{-(2^h-1)} \\
\vdots & \vdots & \ddots & \vdots \\
1 & \zeta_{2^h}^{-(2^h-1)} & \cdots & \zeta_{2^h}^{-(2^h-1)(2^h-1)}
\end{bmatrix}
\begin{bmatrix}K(\chi) \\ K(\chi) \\ \vdots \\ K(\chi)\end{bmatrix} \\
=\begin{bmatrix}2^hK(\chi) \\ 0 \\ \vdots \\ 0\end{bmatrix}\]
And the condition becomes
\begin{align*}
& 2^hK(\chi)\equiv-2^h\pmod{2^{h+1}\cl{P}\bb{Z}[\zeta_{2^hk}]} \\
\Lra\ & K(\chi)\equiv-1\pmod{2\cl{P}\bb{Z}[\zeta_{2^hk}]} \\
\Lra\ & K(\chi)\equiv-1\pmod{2\cl{P}}
\end{align*}
The last step is because $K(\chi)=G(\rho)=\pm p^{vw}$ is an integer under our assumptions.
Finally, note that $K(\chi)\equiv-1\pmod{2\cl{P}}$ if and only if $S(\beta)=0$ by proposition \ref{prop:1}.
\end{proof}

\begin{theorem}
Keep the notations and assumptions as above. Suppose $v'$ is the smallest positive integer such that
$k\mid(p^{v'}+1)$ and $m=2v'w'$. We have
\begin{compactitem}
  \item If $p\equiv1\pmod4$, then $(1+X+\cdots+X^{k-1})^{2^h}\mid S(X)\ \Lra\ w'$ is even.
  \item If $p\equiv3\pmod4$, then $(1+X+\cdots+X^{k-1})^{2^h}\mid S(X)\ \Lra\ w'$ is even or $v'w'$ is odd.
\end{compactitem}
\end{theorem}
\begin{proof}
It simply follows lemma \ref{lem:2} and theorem 4.1 of \cite{Alaca16}.
\end{proof}
\end{section}

\begin{section}{Conclusion}\label{sec:con}
In this paper, we generalize the main result of \cite{Alaca16} by considering the multiplicities of character
values of SLCE sequences. And we apply our results to semiprimitive case and prove that a root must be multiple
root in semiprimitive case.

There are still many open problems to be tackled in this research area.
For example, can we use evaluations of Gauss sums in other cases to get more divisibility results?
Is the necessary condition strong enough to derive a contradiction which implies the multiplicities is small?
Can we derive a better bound for linear complexity of SLCE sequences than \cite{MeidlF206}?
How can we generalize this work to study linear complexity of $d$-ary SLCE sequences?
Is there any method to study $1$-error or even $k$-error linear complexity of SLCE sequences?
\end{section}

\begin{section}{Acknowledgement}
\end{section}

\bibliographystyle{plain}
\bibliography{mybib}

\end{document}